\documentclass[letterpaper, 10 pt, journal, twoside]{ieeetran}

\IEEEoverridecommandlockouts 

\usepackage{amsmath}
\usepackage{amsfonts}
\usepackage{amssymb}
\usepackage{mathrsfs} 
\usepackage{amsthm}

\newtheorem{theorem}{Theorem}
\newtheorem{lemma}[theorem]{Lemma}

\newtheorem{corollary}[theorem]{Corollary}
\newtheorem{proposition}[theorem]{Proposition}

\newtheorem{assumption}{Assumption}

\renewcommand{\ttfamily}{\fontencoding{T1}\fontfamily{lmtt}\selectfont}

\usepackage{booktabs}
\usepackage{romannum}
\usepackage{multirow}
\let\labelindent\relax
\usepackage{enumitem}

\usepackage{graphicx}
\usepackage{tikz}
\usepackage{pgfplots}
\pgfplotsset{compat=newest}
\usepackage{hf-tikz}
\usetikzlibrary{tikzmark}
\usetikzlibrary{calc}

\makeatletter
\let\NAT@parse\undefined
\makeatother
\usepackage[hidelinks]{hyperref}

\usepackage{algorithm}
\usepackage{algpseudocode}

\usepackage{etoolbox}
\errorcontextlines\maxdimen
\newcommand{\ALGtikzmarkcolor}{black}
\newcommand{\ALGtikzmarkextraindent}{4pt}
\newcommand{\ALGtikzmarkverticaloffsetstart}{-.75ex}
\newcommand{\ALGtikzmarkverticaloffsetend}{-.5ex}
\makeatletter
\newcounter{ALG@tikzmark@tempcnta}
\newcommand\ALG@tikzmark@start{%
    \global\let\ALG@tikzmark@last\ALG@tikzmark@starttext%
    \expandafter\edef\csname ALG@tikzmark@\theALG@nested\endcsname{\theALG@tikzmark@tempcnta}%
    \tikzmark{ALG@tikzmark@start@\csname ALG@tikzmark@\theALG@nested\endcsname}%
    \addtocounter{ALG@tikzmark@tempcnta}{1}%
}
\def\ALG@tikzmark@starttext{start}
\newcommand\ALG@tikzmark@end{%
    \ifx\ALG@tikzmark@last\ALG@tikzmark@starttext
    \else
        \tikzmark{ALG@tikzmark@end@\csname ALG@tikzmark@\theALG@nested\endcsname}%
        \tikz[overlay,remember picture] \draw[\ALGtikzmarkcolor] let \p{S}=($(pic cs:ALG@tikzmark@start@\csname ALG@tikzmark@\theALG@nested\endcsname)+(\ALGtikzmarkextraindent,\ALGtikzmarkverticaloffsetstart)$), \p{E}=($(pic cs:ALG@tikzmark@end@\csname ALG@tikzmark@\theALG@nested\endcsname)+(\ALGtikzmarkextraindent,\ALGtikzmarkverticaloffsetend)$) in (\x{S},\y{S})--(\x{S},\y{E});%
    \fi
    \gdef\ALG@tikzmark@last{end}%
}
\apptocmd{\ALG@beginblock}{\ALG@tikzmark@start}{}{\errmessage{failed to patch}}
\pretocmd{\ALG@endblock}{\ALG@tikzmark@end}{}{\errmessage{failed to patch}}
\makeatother

\newcommand{\col}[1]{{\rm col}(#1)}
\newcommand{\image}[1]{{\rm image}#1}
\newcommand{\leftker}[1]{{\rm leftker}#1}
\newcommand{\rank}[1]{{\rm rank}(#1)}
\newcommand{\R}{\mathbb{R}}
\newcommand{\Z}{\mathbb{Z}}
\newcommand{\N}{\mathbb{N}}
\newcommand{\CS}{{\rm io}\mathcal{CS}}
\renewcommand{\H}{\mathscr{H}}
\newcommand{\HH}{\mathcal{H}}
\newcommand{\rz}[1]{#1}
\newcommand{\rzz}[1]{#1}

\usepackage{pdfpages}
\begin{document}


\title{\LARGE \bf
Data-Enabled Predictive Iterative Control
}
\author{Kai Zhang, Riccardo Zuliani, Efe C. Balta, and John Lygeros
\thanks{Research supported by the Swiss National Science Foundation under NCCR Automation (grant agreement 51NF40\_180545). K. Zhang, R. Zuliani and J. Lygeros are with the Automatic Control Laboratory (IfA), ETH Zurich, 8092 Zurich, Switzerland \texttt{\small$\{$zhangkai,rzuliani,lygeros$\}$@ethz.ch}. E. C. Balta is with Inspire AG, 8092 Zurich, Switzerland \texttt{\small efe.balta@inspire.ch}, also with IfA.}}
\maketitle
\thispagestyle{empty}
\pagestyle{empty}


\begin{abstract}
This work introduces the Data-Enabled Predictive iteRative Control (DeePRC) algorithm, a direct data-driven approach for iterative LTI systems. The DeePRC learns from previous iterations to improve its performance and achieves the optimal cost. By utilizing a tube-based variation of the DeePRC scheme, we propose a two-stage approach that enables safe active exploration using a left-kernel-based input disturbance design. This method generates informative trajectories to enrich the historical data, which extends the maximum achievable prediction horizon and leads to faster iteration convergence. In addition, we present an end-to-end formulation of the two-stage approach, integrating the disturbance design procedure into the planning phase. We showcase the effectiveness of the proposed algorithms on a numerical experiment.
\end{abstract}


\section{Introduction}
\emph{Direct data-driven control} enables decision-making directly from raw data, bypassing the need for a parametric model, which received significant attention in the recent literature.
The Fundamental Lemma~\cite{willems_note_2005} establishes that all finite-length trajectories of a linear time-invariant (LTI) system belong to the span of the Hankel matrix constructed using a suitably persistently exciting input/output trajectory, 
and has inspired numerous novel controller designs in the framework of behavioral systems theory. 
For example, \cite{de_persis_formulas_2020} focuses on robustly stabilizing data-driven feedback control design, while 
\cite{coulson_data-enabled_2019,berberich_data-driven_2021} develop a predictive control scheme in a receding-horizon manner. 
The celebrated data-enabled predictive control (DeePC) scheme of \cite{coulson_data-enabled_2019} enables novel applications and theories in the field of direct data-driven control. 
Later extensions of DeePC provide also closed-loop stability guarantees by enforcing the last segment of the predicted trajectory to match the target~\cite{berberich_data-driven_2021}. 

The DeePC scheme uses a Hankel matrix representation of the system dynamics, which is assumed to satisfy a \emph{persistency of excitation condition} (PE condition) \cite{willems_note_2005}.
The degree of excitation of the Hankel matrix limits the maximum prediction horizon of the controller;
hence, a lack of sufficiently rich historical input/output data may produce short prediction horizons, potentially leading to suboptimal performance.
Our approach is to use the measurement data obtained during system execution to improve the PE condition, thus allowing for longer prediction horizons.
This requires persistently exciting inputs which need careful design to ensure exploration without sacrificing performance.
We address this problem in the context of iterative systems by providing a novel scheme with active exploration and theoretical guarantees.

\rz{Iterative systems are those performing repeated tasks, where all individual tasks, called iterations, start from the same initial condition and share the same objective.}
This enables the possibility of iteratively updating the controller using data collected from each iteration.
In the context of iterative systems, \cite{rosolia_learning_2018} proposes a learning Model Predictive Control (LMPC) algorithm that constructs a safe set using past safe trajectories to be used as terminal constraint.
This strategy ensures asymptotic stability and a non-increasing cost over iterations, ultimately converging to the optimal cost under suitable assumptions. 
To improve robustness, \cite{rosolia_robust_2017} extends LMPC to systems with additive uncertainties by leveraging tube-MPC approaches~\cite{mayne_robust_2005}. In this case, robust safe sets are constructed using nominal trajectories satisfying tightened constraints to ensure the nominal cost is non-increasing over iterations.

In this paper, we present the Data-Enabled Predictive iteRative Control (DeePRC), a direct data-driven control approach for iterative LTI systems. 
The DeePRC uses an input/output convex safe set and a terminal cost function designed from previous trajectories and 
enjoys recursive feasibility, asymptotic stability, non-increasing iteration cost, and convergence to the infinite-horizon optimum, albeit asymptotically.
We develop an active exploration variant of DeePRC, which is able to increase its maximum prediction horizon in a sample-efficient way requiring \emph{only a single initial safe trajectory} through the design of an input disturbance term.
We further present an end-to-end formulation of the two-stage approach, integrating the disturbance design procedure into the planning phase.

\paragraph*{Notation}
$[a, b]$ is the set of integers $a$ to $b$. $\|x\|_p$ is the $p$-norm of $x$ and $\|x\|_A^2=x^{\top}Ax$. $\col{x_0,x_1,...,x_\ell}$ and $\col{A_0,A_1,...,A_\ell}$ denote vertically stacked vector $[x_0^{\top}\; x_1^{\top}\; \cdots \; x_\ell^{\top}]^{\top}$ and matrix $[A_0^{\top}\; A_1^{\top}\; \cdots \; A_\ell^{\top}]^{\top}$. For a matrix $A$, $A_{[:m,:]}$ denotes the submatrix of its first $m$ rows, while $A_{[m:,:]}$ represents the remaining rows and we use $A_{[-m:,:]}$ to denote its last $m$ rows.
$\image{(A)}$ and $\leftker{(A)}$ denote the image and left kernel of matrix $A$, respectively. 
For a generic length-$N$ sequence $\{y_k\}_{k=0}^{N-1}$, we write $y_{[a,b]} := \col{y_a, \ldots, y_b}$.

\section{Preliminaries} \label{sec:preliminaries}
\subsection{Problem Setting}\label{subsection:problem_setting}
Consider the unknown discrete-time \rz{LTI} system
\begin{equation}
\label{eq:ABCD}
    \begin{aligned}
        x_{t+1} & = Ax_{t} + Bu_{t}, \\
        y_{t} & = Cx_{t} + Du_{t}, 
    \end{aligned}
\end{equation}
where $x_{t} \in \mathbb{R}^{n}$ is the state, $u_{t} \in \mathbb{R}^{m}$ is the input, and $y_{t} \in \mathbb{R}^{p}$ is the output.
The state $x_t$ is not directly measurable, and needs to be reconstructed using input/output information. \rz{Denote by $\underline{\ell}\in \mathbb{N}_{> 0}$ the lag of the system, that is, the minimum integer such that the observability matrix $\mathscr{O}_{\underline{\ell}}(A,C) := \col{C,CA,...,CA^{\underline{\ell}-1}}$ \rzz{has rank $n$}}.
We assume that an upper bound $\ell$ on the lag $\underline{\ell}$ is available.

To represent the state using input/output data, we define the extended state following~\cite[Definition~3]{berberich_design_2021}:
\begin{equation}
    \xi_t := \col{u_{[t-\ell,t-1]} , y_{[t-\ell,t-1]}}
    \in \mathbb{R}^{n_{\xi}},
\label{eq:extend_state}
\end{equation}
where $n_{\xi} = (m+p)\ell$. Here, $\xi_t$ is the state of a non-minimal representation of \eqref{eq:ABCD} with dynamics
\begin{equation}
\label{eq:ABCD_tilde}
    \begin{aligned}
        \xi_{t+1} & = \tilde{A}\xi_{t} + \tilde{B}u_{t}, \\
        y_{t} & = \tilde{C}\xi_{t} + \tilde{D}u_{t} ,
    \end{aligned}
\end{equation}
which is generally not unique. We refer to \cite[Equation~6]{berberich_design_2021} for one possible form of~\eqref{eq:ABCD_tilde}. 

Suppose system \eqref{eq:ABCD} performs iterative tasks, where the task for each iteration is to steer the system from a starting \rzz{equilibrium} $\xi^{\rm S}$ to a target \rzz{equilibrium} $\xi^{\rm F}$.
We define $ u^j = \col{u^j_{-\ell},\ldots,u^j_{-1},u^j_0,\ldots}$ and $y^j = \col{y^j_{-\ell},\ldots,y^j_{-1},y^j_0,\ldots}$, where $u^j_t = u^{\rm S}$, $y^j_t = y^{\rm S}$, $\forall t \leq 0$ describe the initial state for the associated extended state sequence $\{\xi^j_t\}_{t=0}^{T^j}$, with $\xi^j_0 = \xi^{\rm S}$.

A trajectory is considered \textit{safe} if it converges to the target state $\xi^{\rm F}:=\operatorname{col}(u^{\rm F},\dots,u^{\rm F},y^{\rm F},\dots,y^{\rm F})$ and satisfies the convex input and output constraints
\begin{align}\label{eq:constraints}
u^j_t \in \mathcal{U},~ y^j_t \in \mathcal{Y},
\end{align}
for all $t \in \Z_{[0,T^j-1]}$, where $\mathcal{U}\subseteq\R^m$ and $\mathcal{Y}\subseteq\R^p$ are convex sets.
Following \cite{rosolia_learning_2018}, the control objective is to solve the following infinite horizon problem.
\begin{subequations}
\label{eq:ift_prob}
\begin{align}
    J_{\infty}^{*}(\xi^{\rm S}) = 
    \min_{u_0,u_1,\ldots} \quad & \sum_{k=0}^{\infty}{h(y_k,u_k)} \label{eq:ift_prob:cost} \\
    {\rm s.t.} \quad~\, & \xi_{k+1} = \Tilde{A} \xi_k + \Tilde{B} u_k, \ \forall k \geq 0, \label{eq:ift_prob:dynamics} \\ 
    & y_k = \Tilde{C} \xi_k + \Tilde{D} u_k, \ \forall k \geq 0, \label{eq:ift_prob:measurement} \\ 
    & \xi_0 = \xi^{\rm S}, \label{eq:ift_prob:ini_cond} \\ 
    & u_{k} \in \mathcal{U}, \ y_{k} \in \mathcal{Y}, \ \forall k \geq 0. \label{eq:ift_prob:constr}
\end{align}
\end{subequations}
The stage cost is $h(u_k,y_k) = \|u_k - u^{\rm F}\|_{R}^{2} + \|y_k - y^{\rm F}\|_{Q}^2$ with \rz{$Q,R \succ 0$}, \rzz{and we assume $(u^\mathrm{F},y^\mathrm{F})\in \operatorname{int} (\mathcal{U} \times \mathcal{Y})$}.
\begin{assumption}\label{ass:stand}
	System~\eqref{eq:ABCD} is observable and stabilizable.
\end{assumption}

Problem \eqref{eq:ift_prob} is not tractable due to the infinite horizon and the lack of knowledge of the system dynamics.
We develop the DeePRC algorithm as a tractable solution to \eqref{eq:ift_prob}. Our scheme uses a data-based non-parametric model for \eqref{eq:ift_prob:dynamics}-\eqref{eq:ift_prob:measurement} to perform receding horizon control.

\subsection{Review of DeePC}
Here, we provide a brief background on the Data Enabled Predictive Control (DeePC).
Let $T,L \in \mathbb{Z}_{\geq 0}$ with $T \geq L > \ell$. 
The Hankel matrix of depth $L$ for a signal $u = \col{u_0, \ldots,u_T}$ is defined as
\begin{equation*}
    \mathscr{H}_{L}(u) := 
    \begin{bmatrix}
        u_0 & u_1 & \cdots & u_{T-L} \\
        u_1 & u_2 & \cdots & u_{T-L+1} \\
        \vdots & \vdots & \ddots & \vdots \\
        u_{L-1} & u_{L} & \cdots & u_{T-1}
    \end{bmatrix}.
\end{equation*}
We use the identifiability condition presented in~\cite{markovsky_identifiability_2023} for theoretical guarantees, which is a slight generalization of the Fundamental Lemma~\cite{willems_note_2005}. 
\begin{lemma}[{\cite[Corollary~21]{markovsky_identifiability_2023}}] \label{lemma:identifiability_condition}
    Consider a length $T$ input/output trajectory $\{u^{\rm d}, y^{\rm d}\}$ of system \eqref{eq:ABCD}. The image of $\mathscr{H}_{L}(u^{\rm d},y^{\rm d}):= \col{\mathscr{H}_{L}(u^{\rm d}), \mathscr{H}_{L}(y^{\rm d})}$ is the span of all length $L$ trajectories of the system if and only if 
\begin{equation}
    \rank{\mathscr{H}_{L}(u^{\rm d},y^{\rm d})} = mL + n.
\label{eq:pe_condition}
\end{equation}
\end{lemma}

\smallskip

\rz{Under condition \eqref{eq:pe_condition}, one could identify a state-space (using standard system identification techniques) or \rzz{kernel} representation (following \cite{alsalti2023sample}) of the system. The main advantage of a direct data-driven approach, like the one presented in this paper, is its simplicity and efficiency, as well as the possibility to directly deal with noisy data by introducing regularization terms in the cost function. For a detailed discussion about the differences between direct and indirect methods, we refer the reader to \cite{dorfler_bridging_2023}. A direct comparison between our approach and that of \cite{alsalti2023sample} is subject of future research.}
Given an initial horizon $T_{\rm ini} \in \N_{> 0}$ and a prediction horizon $N \in \N_{> 0}$, let
\begin{equation}
    \begin{bmatrix}U_{\rm P} \\ U_{\rm F}\end{bmatrix} := \mathscr{H}_{T_{\rm ini}+N}(u^{\rm d}),
    \quad
    \begin{bmatrix}Y_{\rm P} \\ Y_{\rm F}\end{bmatrix} := \mathscr{H}_{T_{\rm ini}+N}(y^{\rm d}),
\label{eq:up_yp_uf_yf}
\end{equation}
where $U_{\rm P}~\in~\mathbb{R}^{m T_{\rm ini} \times H}$, $U_{\rm F}~\in~\mathbb{R}^{mN \times H}$, $Y_{\rm P}~\in~\mathbb{R}^{p T_{\rm ini} \times H}$, $Y_{\rm F}~\in~\mathbb{R}^{pN \times H}$ with $H = T - T_{\rm ini} - N + 1$. Suppose that $(u^\text{d},y^\text{d})$ satisfy the rank condition \eqref{eq:pe_condition} with $L=T_{\rm ini}+N$.
According to Lemma \ref{lemma:identifiability_condition}, $\{\col{u_{\rm ini},u}, \col{y_{\rm ini},y}\}$ is a valid trajectory of the system if and only if $\exists g \in \mathbb{R}^{H}$ such that
\begin{equation}
    \begin{bmatrix}U_{\rm P} \\ Y_{\rm P} \\ U_{\rm F} \\ Y_{\rm F}\end{bmatrix} g =
    \begin{bmatrix}u_{\rm ini} \\ y_{\rm ini} \\ u \\ y\end{bmatrix} ,
\label{eq:non-param_model}
\end{equation}
where $u_{\rm ini} \in \mathbb{R}^{mT_{\rm ini}}$ and $y_{\rm ini} \in \mathbb{R}^{pT_{\rm ini}}$ are the $T_{\rm ini}$-step past input/output observations.
If $T_{\rm ini} \geq \underline{\ell}$, given any future control input sequence $u \in \mathbb{R}^{mN}$, the associated output sequence $y \in \mathbb{R}^{pN}$ is uniquely determined via \eqref{eq:non-param_model} thanks to Lemma \ref{lemma:identifiability_condition}. In the remainder of this paper, we assume $T_\text{ini}\!=\ell$.

\section{Nominal DeePRC} \label{sec:nominal_deeprc}
In this section, we introduce the nominal DeePRC algorithm and exploration schemes using past trajectory data.

\subsection{Input/Output Convex Safe Set}
The input/output convex safe set $\CS^j$ at iteration $j$ is defined as the convex hull of all \emph{extended states} along previous safe trajectories:
\begin{equation}
    {\rm io}\mathcal{CS}^{j} := \left\{
        \sum_{i \in M^{j}}{
            \sum_{t=0}^{\infty}{
                \gamma_{t}^{i} \xi^i_{t}
            }
        } \mid
        \gamma_{t}^{i} \geq 0,
        \sum_{i \in M^{j}}{\sum_{t=0}^{\infty}{\gamma_{t}^{i}}} = 1
    \right\} ,
\label{eq:ioCS}
\end{equation} 
where $M^{j} := \{i \in \Z_{[0,j-1]} \mid \rzz{\sum_{t=0}^\infty h(y_t^i,u_t^i) < \infty}, u_t^i\in\mathcal{U}, y_t^i\in\mathcal{Y},\rz{\forall t}\}$ indicates past safe iterations. \rzz{Finite trajectories can be extended by appending $(y^{\rm F},u^{\rm F})$}.
For all $i\leq j$, $M^{i} \subseteq M^{j}$, and therefore ${\rm io}\mathcal{CS}^{i} \subseteq {\rm io}\mathcal{CS}^{j}$.
Every $\xi\in\CS^j$ is safe, meaning that it can be steered to $\xi^{\rm F}$ without violating the constraints.
\rz{Note that $\CS^j$ is updated only in between iterations.}
The \textit{cost-to-go} of $\xi_t^i$ is defined as
\begin{equation*}
    J^i(\xi^{i}_{t}) = \sum_{k=t}^{\infty}{h(u_k^i,y_k^i)}.
\end{equation*}
representing the cost required to steer $\xi_t^i$ to $\xi^{\rm F}$.
Note that $J^i(\xi_{0}^{i}) = J^i(\xi^{\rm S})$ is the cost of the entire trajectory $i$.
The cost-to-go can be extended to any $\xi \in {\rm io}\mathcal{CS}^{j}$, by defining
\begin{equation}
\begin{aligned}
    P^{j}(\xi) = 
    \min_{\gamma \geq 0} \quad & \sum_{i \in M^{j}}{\sum_{t=0}^{\infty}{\gamma_{t}^{i} J^i(\xi_{t}^{i})}} \\
    {\rm s.t.} \quad & \sum_{i \in M^{j}}{\sum_{t=0}^{\infty}{\gamma_{t}^{i}}} = 1, \ 
    \sum_{i \in M^{j}}{\sum_{t=0}^{\infty}{\gamma_{t}^{i} \xi^i_{t}}} = \xi.
\end{aligned}
\label{eq:terminal_cost}
\end{equation}
Note that $P^j(\xi)$ represents the cost required to safely steer $\xi$ to $\xi^{\rm F}$ under the control policy $\pi(\xi)=\sum_{i \in M^{j}}{\sum_{t=0}^{\infty}{\gamma_{t}^{i,*} u_{t}^i}}$, where $\gamma^{*}$ is the minimizer of \eqref{eq:terminal_cost}.

\begin{theorem} \label{theorem:invariance}
The set ${\rm io}\mathcal{CS}^{j}$ is control invariant\footnote{We refer the reader to \cite{borrelli_constrained_2003} for an exact definition.} for all $j$. For any $\xi \in {\rm io}\mathcal{CS}^j$, $P^{j}(\xi^{+}) \leq P^{j}(\xi) - h(\bar{u},\bar{y})$, $\bar{u} \in \mathcal{U}$ and $\bar{y} \in \mathcal{Y}$, where $\xi^+ = \tilde{A} \xi + \tilde{B} \bar{u}$, $\bar{y} = \tilde{C} \xi + \tilde{D} \bar{u}$, and $\textstyle{\bar{u} = \pi(\xi)}$.
\end{theorem}

\begin{proof}
For any $\xi \in {\rm io}\mathcal{CS}^{j}$, let
\begin{align*}
\bar{u}&:=\textstyle\sum_{i\in M^j} \sum_{t=0}^{\infty} \gamma_t^{i,*}u_t^i, ~ \bar{y}:=\textstyle\sum_{i\in M^j} \sum_{t=0}^{\infty} \gamma_t^{i,*}y_t^i,
\end{align*}
where $\bar{u}\in \mathcal{U}$ and $\bar{y}\in \mathcal{Y}$ \rz{because of convexity}, and $\gamma^*$ is a minimizer of \eqref{eq:terminal_cost}. \rz{By linearity} we have $\xi^+:=\tilde{A}\xi+\tilde{B}\bar{u}\in\CS^j$ and, recognizing that $\tilde{\gamma}$, defined for $t\geq 1$ as $\tilde{\gamma}_t^i=\gamma_{t-1}^{i,*}$, is a feasible solution of \eqref{eq:terminal_cost} with $\xi=\xi^+$, we have
\begin{align*}
    P^{j}(\xi^{+}) & \leq 
   \textstyle \sum_{i \in M^{j}}{
        \textstyle\sum_{t=0}^{\infty}{
            \tilde{\gamma}_{t+1}^i J^i(\xi_{t+1}^{i})
        }
    } \\
    & = \textstyle\sum_{i \in M^{j}}{
        \textstyle\sum_{t=0}^{\infty}{
            \gamma_{t}^{i,*} \left(J^i(\xi_{t}^{i}) - h(u_{t}^{i},y_{t}^{i}) \right)
        }
    } \\
    & \leq P^{j}(\xi) - h(\bar{u},\bar{y}),
\end{align*}
\rz{where the last inequality follows from convexity.}
\end{proof}

\subsection{The Nominal DeePRC Algorithm}
The nominal DeePRC algorithm utilizes a single safe trajectory $\{u^0,y^0\}$ to construct both the Hankel matrix and the safe set $\CS^1$. Since the closed-loop trajectories can potentially have infinite length, we assume only a finite portion of them is used to construct the non parametric representation of the system. We use $T^j$ to denote the length of the recorded trajectory.

\begin{assumption} \label{assumption:ini_safe_traj}
\rzz{A safe trajectory $\{u^0,y^0\}$ is available with $u^0_{[0,T^0]}$, $y^0_{[0,T^0]}$ satisfying \eqref{eq:pe_condition} with $\textstyle{L=T_{\rm ini}\!+\!N}$ and $\textstyle{N\geq\ell}$.}
\end{assumption}

Let $\operatorname{col}(U_P,Y_P,U_F,Y_F)=\H_{T_{\text{ini}}+N}(u^0,y^0)$, $u(t)\in\R^{mN}$, $y(t)\in\R^{pN}$, $g(t)\in\R^{H}$, and
\begin{align*}
J_N^j(u(t),y(t),\xi_N(t)) := \sum_{k=0}^{N-1}{h(u_{k}(t),y_{k}(t))} + P^{j}(\xi_{N}(t)).
\end{align*}
The nominal DeePRC for time-step $t$ and iteration $j$ is formulated as follows:
\begin{subequations}
\label{eq:nom_deeprc}
\begin{align}
&J^{j,*}_N(\xi_t^j)\! = \!\operatorname*{min}_{g(t), u(t), y(t)} J_N^j(u(t),y(t),\xi_N(t)) \label{eq:nom_deeprc:cost} \\
& {\rm s.t.} ~ \begin{bmatrix} U_{\rm P}\\ Y_{\rm P}\\ U_{\rm F}\\ Y_{\rm F}\end{bmatrix} g(t) = \begin{bmatrix} u_{\rm ini}(t)\\ y_{\rm ini}(t)\\ u(t)\\ y(t)\end{bmatrix}, ~ \begin{bmatrix} u_{\rm ini}(t) \\ y_{\rm ini}(t)\end{bmatrix} = \xi^{j}_{t}, \label{eq:nom_deeprc:model}\\
& \hphantom{\rm s.t. ~} \xi_{N}(t) = \begin{bmatrix}u_{[N-\ell,N-1]}(t) \\ y_{[N-\ell,N-1]}(t)\end{bmatrix}, ~ 
\xi_{N}(t) \in {\rm io}\mathcal{CS}^{j}, \label{eq:nom_deeprc:terminal_cnstr}\\
& \hphantom{\rm s.t. ~} u_{k}(t) \in \mathcal{U}, ~ y_{k}(t) \in \mathcal{Y}, ~ \forall k \in [0,N-1],\label{eq:nom_deeprc:io_cnstr}
\end{align}
\end{subequations}

\subsection{Hankel Matrix Update}\label{subsection:hankel_update}
The DeePRC \eqref{eq:nom_deeprc} utilizes a static Hankel matrix. As new iterations $\{u^j,y^j\}$ are measured, we can augment the prediction horizon by modifying the Hankel matrix.
At iteration $j$, the largest feasible prediction horizon is
\begin{equation}
   N^j = \max \left\{
       N \, \middle\vert \, \rank{\mathcal{H}^j_N} = m(\ell+N) + n
   \right\} ,
\label{eq:max_N}
\end{equation}
with $\mathcal{H}^j_N = \begin{bmatrix} \mathscr{H}_{\ell+N}(u^0,y^0) & \ldots & \mathscr{H}_{\ell+N}(u^{j-1},y^{j-1}) \end{bmatrix}$.
Suppose our goal is \rz{to have $N^j=\bar{N} \geq \ell$ for some $j$}. \rz{The desired horizon $\bar{N}$ is a design parameter that can be chosen before the beginning of the operation, or can be incrementally updated online until some stopping criterion is met. We assume that $\bar{N}<T^j$ for all $j$.}
To fulfill the objective we require $\HH_{\bar{N}}^j$ to satisfy the rank condition \eqref{eq:pe_condition} with $L=\bar{L} := \ell + \bar{N}$.
If $\bar{N}>N^j$, then we solve \eqref{eq:nom_deeprc} with prediction horizon $N^j$ and matrix $\HH_{N^j}^j$.
The process continues until $\rank{\HH_{\bar{N}}^{\bar{j}}} = m\bar{L} + n$  for some $\bar{j}$, after which we initialize \eqref{eq:nom_deeprc} with $\bar{N}$ and \rz{$\HH_{\bar{N}}^{\bar{j}+1}$}.

\rz{This type of update takes place at the end of each iteration, however, we can also study how each new sample, gathered within the iteration, affects the rank of $\mathcal{H}_{\bar{N}}^j$. To this end, let $\HH_{\bar{N}}^{j,t}\!=\HH_{\bar{N}}^{j}$ for all $t<\bar{N}-1$ and define for all $t \geq \bar{N}-1$}
\begin{equation}\label{eq:max_hankel_update}
\HH_{\bar{N}}^{j,t+1}=
    \begin{bmatrix}
        \HH_{\bar{N}}^{j,t} &
        \begin{bmatrix}
            u^j_{[t-\bar{L}+1,t]} \\ y^j_{[t-\bar{L}+1,t]}
        \end{bmatrix}
    \end{bmatrix}.
\end{equation}
\rz{Ideally, each update in \eqref{eq:max_hankel_update} should ensure that $\operatorname{rank} \mathcal{H}_{\bar{N}}^{j,t+1} = 1 + \operatorname{rank} \mathcal{H}_{\bar{N}}^{j,t}$. Note that the matrices $\mathcal{H}_{\bar{N}}^{j,t}$ obtained through \eqref{eq:max_hankel_update} are not used in the DeePRC problem, they are only used to understand how newly gathered samples affect the rank of the Hankel matrix of depth $\bar{N}$, and will later be used for the input disturbance design.}

\section{Safe and Active Exploration} \label{sec:active_exploration}
Update \eqref{eq:max_hankel_update} can be inefficient, as new trajectories are not guaranteed to increase the rank condition and extend the prediction horizon.
In this section, we present an \emph{active exploration} strategy to safely reach the desired prediction horizon $\bar{N}$.
We design the control input as the sum of two components $u_t = \tilde{u}_t + d_t$, where $\tilde{u}_t$ is chosen to optimize closed-loop performance, and $d_t$ is an artificial input disturbance ensuring active exploration.

\subsection{Safe Exploration: Tube DeePRC}
To ensure safety, we assume $d_t$ is chosen from a known, user-designed bounded set $\mathcal{D}$, and develop a tube-based DeePRC scheme given by
\begin{subequations}
\label{eq:deeprc_2s}
\begin{align}
    \min_{g(t), v(t), z(t)} & ~ J^j_N(v(t),z(t),\zeta_{N}(t)) \\
    {\rm s.t.} \quad~ & \left(g(t),v(t),z(t),\zeta_{N}(t)\right) \in \mathcal{C}, \label{eq:deeprc_2s:compact_cnstr}\\
    & v_{k}(t) \in \bar{\mathcal{U}}, \ z_{k}(t) \in \bar{\mathcal{Y}}, \ \forall k \in [0,N-1], \label{eq:deeprc_2s:io_cnstr}
\end{align}
\end{subequations}
where $\mathcal{C}$ denotes the constraints \eqref{eq:nom_deeprc:model}-\eqref{eq:nom_deeprc:terminal_cnstr} with $(u,y,\xi)$ replaced with $(v,z,\zeta)$, and $z_t$, $v_t$, and $\zeta_t$ are the nominal output, input, and extended state, respectively.
The tightened constraints $\bar{\mathcal{U}} = \mathcal{U} \ominus \mathcal{E}_{u}$ and $\bar{\mathcal{Y}} = \mathcal{Y} \ominus \mathcal{E}_{y}$ are obtained using a robust positive invariant set $\mathcal{E}$ by $\mathcal{E}_{u} = \{e_u \in \mathbb{R}^{m} \mid \exists e \in \mathcal{E}, \ {\rm s.t.} \ e_u = T_{u} e\}$ and $\mathcal{E}_{y} = \{e_y \in \mathbb{R}^{p} \mid \exists e \in \mathcal{E}, \ {\rm s.t.} \ e_y = T_{y} e\}$, with $T_u$ and $T_y$ selecting the entries associated to the last input and the last output of the error $e$ between the true and the nominal extended state.
We then choose
\begin{equation}
    \tilde{u}^j_t = K(\xi^j_t - \zeta^j_t) + v_0^*(t),
\label{eq:tube_deeprc_control_law}
\end{equation}
where $K$ is a stabilizing output feedback gain (\rz{which can be computed either using \cite[Theorem 8]{de_persis_formulas_2020} or, under less restrictive conditions, using the techniques of \cite{alsalti2023notes}}).

One way to obtain $\mathcal{E}$ is to consider the nominal system
\begin{align*}
    \zeta_{t+1} & = \tilde{A}\zeta_{t} + \tilde{B}v_{t} \\
    z_{t} & = \tilde{C}\zeta_{t} + \tilde{D}v_{t} .
\end{align*}
and obtain $(\tilde{A},\tilde{B})$ from \eqref{eq:non-param_model} with $N=1$ as
\begin{equation*}
    \xi_{t+1} \! = \! \begin{bmatrix}
        \col{U_{{\rm P}[m:,:]}, U_{\rm F}} \\
        \col{Y_{{\rm P}[p:,:]}, Y_{\rm F}}
    \end{bmatrix}
    \begin{bmatrix}U_{\rm P} \\ Y_{\rm P} \\ U_{\rm F}\end{bmatrix}^{\dagger}
    \begin{bmatrix}u_{\rm ini}(t) \\ y_{\rm ini}(t) \\ u(t) \end{bmatrix}
    \! = \! \tilde{A} \xi_t + \tilde{B} u_t,
\end{equation*}
\rz{which always admits a solution under Assumption \ref{assumption:ini_safe_traj} thanks to \cite[Theorem 1]{de_persis_formulas_2020}}.
Assuming $(\tilde{A},\tilde{B})$ is stabilizable (sufficient conditions for stabilizability have been studied in \cite[Theorem 3]{de_persis_formulas_2020}, \rz{alternatively, one can perform a change of coordinates to obtain the extended-state of \cite{alsalti2023notes}}) we can choose $\mathcal{E}$ (either polytopic or ellipsoidal) to satisfy
\begin{equation}
e_{t+1} = (\tilde{A} + \tilde{B} K)e_t + \tilde{B} d_t \in \mathcal{E}, \ \forall e_t \in \mathcal{E}, \ \forall d_t \in \mathcal{D},
\label{eq:rpi}
\end{equation}
where $e_t:=\xi_t-\zeta_t$.
\rz{One practical way to avoid the $(\tilde{A},\tilde{B})$ parameterization, at the expense of a potential violation of the theoretical guarantees, is to choose $\mathcal{E}$ to be a sufficiently large ball centered at the origin. Data-driven techniques to estimate $\mathcal{E}$ are an active area of research and represent an interesting direction for future work.}

For safe planning, any element in the $\CS$ of problem \eqref{algo:deeprc_2s} is required to meet the tightened constraints $\bar{\mathcal{U}}, \bar{\mathcal{Y}}$. As in \cite{rosolia_robust_2017}, we need to strengthen Assumption \ref{assumption:ini_safe_traj} as follows.
\begin{assumption} \label{assumption:ini_robustly_safe_traj}
\rzz{A robustly safe trajectory $\{v^0,z^0\}$ (i.e., with $v^0_t \in \bar{\mathcal{U}}, z^0_t \in \bar{\mathcal{Y}}, \forall t \in \mathbb{Z}_{\geq 0}$) is available with $v^0_{[0,T^0]}$, $z^0_{[0,T^0]}$ satisfying \eqref{eq:pe_condition} with $\textstyle{L=T_{\rm ini}\!+\!N}$ and $\textstyle{N\geq\ell}$.}
\end{assumption}

To ensure robust constraint satisfaction, we construct the $\CS$ in \eqref{eq:deeprc_2s} using the nominal trajectories $\{v^j,z^j\}$. 

\subsection{Active Exploration: LKB Input Disturbance Design}\label{subsection:LKB}
Since $\tilde{u}_t^j$ in \eqref{eq:tube_deeprc_control_law} satisfies the nominal constraints for all $d_t^j\in\mathcal{D}$ \rz{(compare Corollary \ref{corollary:robust_deeprc_rec_feas})}, we can design $d_t^j$ to ensure the Hankel matrix update \eqref{eq:max_hankel_update} always produces an increment in the rank of the matrix. 
Choosing $d_t^j$ independently and uniformly from $\mathcal{D}$ at each time step may be inefficient, since the control input $\tilde{u}_t^j$ may already contain sufficient information to excite the system.
We propose a left-kernel-based (LKB) disturbance design inspired by \cite{van_waarde_beyond_2022} as a sample-efficient exploration strategy.

For the first $N^{j}-2$ time-steps of iteration $j$ we set $d_t^j=0$. Next, for $t \geq N^{j}-1$, let $\operatorname{col}(\bar{U}^{j,t},\bar{Y}^{j,t}) :=\HH_{\bar{N}}^{j,t}$, where $\bar{U}^{j,t} \in \mathbb{R}^{m\bar{L} \times H_c}$, $\bar{Y}^{j,t} \in \mathbb{R}^{p\bar{L} \times H_c}$ with $H_c = t-\bar{N}+1 + \sum_{i=0}^{j-1}{\left(T^i-\bar{N}+1\right)}$, and let $\bar{Y}^{j,t}_{\rm up} = \bar{Y}^{j,t}_{[:p(\bar{L}-1),:]}$.
If
\begin{equation}
    \col{y^j_{[t-\bar{L}+1,t-1]} , u^j_{[t-\bar{L}+1,t-1]} , \tilde{u}^j_t}
    \notin
    \image{\begin{bmatrix}
        \bar{Y}^{j,t}_{\rm up} \\ \bar{U}^{j,t}
    \end{bmatrix}},
\label{eq:excitation_check}
\end{equation}
then the candidate control input $\tilde{u}^j_t$ will excite the trajectory and thus we set $d^j_t = 0$.
Otherwise, \rz{consider any} $\kappa\in \mathbb{R}^{p(\bar{L}-1)+m\bar{L}}$ in the left kernel of $\operatorname{col}(\bar{Y}^{j,t}_{\rm up},\bar{U}^{j,t})$
with $\kappa_{[-m:]} \neq 0$ and take any $d_t^j$ satisfying
\begin{equation}
    \kappa^{\top} \col{y^j_{[t-\bar{L}+1,t-1]} , u^j_{[t-\bar{L}+1,t-1]} , \tilde{u}^j_t + d_t^j}
    \neq 0,
\label{eq:d_design_case2}
\end{equation}
where $\kappa_{[-m:]}$ denotes the last $m$ elements of $\kappa$~\cite{van_waarde_beyond_2022}.
If $\mathcal{D}$ is absorbing (e.g. any norm ball), there always exists some $d_t^j$ satisfying \eqref{eq:d_design_case2}. We then apply the input $u_t^j=\tilde{u}_t^j+d_t^j$ with $\tilde{u}_t^j$ as in \eqref{eq:tube_deeprc_control_law}. This ensures that $\rank{\HH_{\bar{N}}^{j,t+1}} = \rank{\HH_{\bar{N}}^{j,t}} + 1$ for all $t\geq N^{j}-1$, meaning that the desired horizon $\bar{N}$ can be reached in exactly $m\bar{L}+n - \rank{\HH^{1}_{\bar{N}}}$ time steps. Afterwards, $d_t^j\equiv 0$ and the nominal DeePRC \eqref{eq:nom_deeprc} can be used instead of \eqref{eq:deeprc_2s}.
The DeePRC scheme with LKB active exploration is summarized in Algorithm \ref{algo:deeprc_2s}.

\subsection{An End-to-End Formulation}
Algorithm \ref{algo:deeprc_2s} presents a two-stage approach for designing the control input. 
We also provide an end-to-end formulation, which integrates the online disturbance design procedure into the optimization problem
\begin{subequations}
\label{eq:deeprc_1s}
\begin{align}
    & \min_{g(t), v(t), z(t), d(t)} J_N^j(v(t),z(t),\zeta_{N}(t))\! +\! \lambda \|d(t)\|_1  \label{eq:deeprc_1s:cost} \\
    & \qquad~~ {\rm s.t.} ~~ \left(g(t),v(t),z(t),\zeta_{N}(t)\right) \in \mathcal{C}, \label{eq:deeprc_1s:compact_cnstr}\\
    & \hphantom{\qquad~~ {\rm s.t.} ~~} v_{k}(t) \in \bar{\mathcal{U}}, \ z_{k}(t) \in \bar{\mathcal{Y}}, \ \forall k \in [0,N-1], \label{eq:deeprc_1s:io_cnstr} \\
    & \hphantom{\qquad~~ {\rm s.t.} ~~} u_0(t) = v_0(t) + K(\xi^j_t - \zeta^j_t), \label{eq:deeprc_1s:cand_input}\\
    & \hphantom{\qquad~~ {\rm s.t.} ~~} \|\begin{bmatrix}y^j_{[t-\bar{L}+1,t-1]} \\ u^j_{[t-\bar{L}+1,t-1]} \\ u_0(t) + d(t)\end{bmatrix}^{\top} \!\!\! \mathcal{K}^{j,t}\|_{\infty} \geq \epsilon, d(t) \in \mathcal{D} \label{eq:deeprc_1s:excitation_condition}.
\end{align}
\end{subequations}
Condition \eqref{eq:excitation_check} is here enforced through the excitation constraint \eqref{eq:deeprc_1s:excitation_condition}, where $\epsilon \in \mathbb{R}_{>0}$ and $\mathcal{K}^{j,t}$ is an orthogonal basis of $\operatorname{col}(\bar{Y}^{j,t}_{\rm up}, \bar{U}^{j,t})$.
We further add $1$-norm penalty on the disturbance to ensure that its magnitude is kept to a minimum.
\eqref{eq:deeprc_1s} can be formulated as a mixed-integer quadratic program (MIQP).
The \rz{end-to-end} formulation should generally perform better than the two-stage, as $d(t)$ is part of the optimization problem.
Similar to the two-stage approach, we only solve \eqref{eq:deeprc_1s} until $\bar{N}$ is reached, then we utilize \eqref{eq:nom_deeprc}.

\newpage

\vphantom{0pt}\\[-2.8em]
\begin{algorithm}[H]
    \caption{DeePRC with LKB Active Exploration}\label{algo:deeprc_2s}
    \begin{algorithmic}[1]
    \State \textbf{Init}: $\{u^0,y^0\}$, $\xi^{\rm S}$, $\xi^{\rm F}$, $\bar{N} \geq \ell$, $\mathcal{D}$, $\CS^1$, $j=1$.
    \While{$ y^j\neq y^{j-1} $}
        \State $t=0$, $\xi^{j}_{0} = \xi^{\rm S}$, $\zeta^{j}_{0} = \xi^{\rm S}$.
        \If{$\bar{N}$ has been reached}
            \State Setup \eqref{eq:nom_deeprc} using $\mathcal{H}_{\bar{N}}^{\bar{j}+1}$.
            \While{not converged to $\xi^{\rm F}$}
                \State Solve \eqref{eq:nom_deeprc} and apply $u_t^j$. Get $y_{t}^{j}$ and $\xi_{t+1}^{j}$.
            \EndWhile
            \State ${\rm io}\mathcal{CS}^{j+1} = {\rm io}\mathcal{CS}^{j} \cup \{u^j,y^j\}$.
        \Else
            \State Compute $N^j$ using \eqref{eq:max_N}.
            \State Setup \eqref{eq:deeprc_2s} using $\HH_{N_j}^j$.
            \While{not converged to $\xi^{\rm F}$}
                \State Solve \eqref{eq:deeprc_2s} and get $\tilde{u}_t^j$ from \eqref{eq:tube_deeprc_control_law}. 
                \If{$t \leq \bar{N}\!-\!2$ or \eqref{eq:excitation_check} holds}
                    \State $d^j_t=0$.
                \Else
                    \State Select any $d^j_t$ satisfying \eqref{eq:d_design_case2}.
                    \State Update $\HH_{\bar{N}}^{j,t}$ with \eqref{eq:max_hankel_update}.
                \EndIf
                \State Apply $u^j_t = \tilde{u}^j_t + d^j_t$. Get $y_{t}^{j}$ and $\xi_{t+1}^{j}$.
            \EndWhile
            \State ${\rm io}\mathcal{CS}^{j+1} = {\rm io}\mathcal{CS}^{j} \cup \{v^j,z^j\}$.
        \EndIf
        \State $j\gets j+1$.
    \EndWhile
    \end{algorithmic}
\end{algorithm}

\section{Properties and Guarantees} \label{sec:properties}
\subsection{Closed-Loop Properties}
We first prove that the scheme in \eqref{eq:nom_deeprc} enjoys recursive feasibility and stability. We require the following.
\begin{assumption}\label{assumption:quad_upperbound}
\rzz{There exists $c_u>0$ such that $J_N^{j,*}(\xi) \leq c_u \|\xi\|_2^2$ for all $\xi$ for which \eqref{eq:nom_deeprc} is feasible.}
\end{assumption}

\begin{proposition} \label{proposition:closed-loop}
    Under Assumptions \ref{ass:stand}, \ref{assumption:ini_safe_traj}, and \ref{assumption:quad_upperbound}, the following hold for all $j \geq 1$:
    \textit{\romannum{1}}) Problem \eqref{eq:nom_deeprc} is feasible for all $t \in \mathbb{Z}_{\geq 0}$;
    \textit{\romannum{2}}) $(u^j,y^j)$ satisfy the constraints $\mathcal{U}$, $\mathcal{Y}$;
    \textit{\romannum{3}}) $\xi_t^j$ \rzz{exponentially} converges to $\xi^{\rm F}$.
\end{proposition}
\begin{proof}
    Assumption \ref{assumption:ini_safe_traj} ensures $\xi^0_{\tau} \in {\rm io}\mathcal{CS}^j, \forall \tau \in [0,T^0-1]$. At $t=0$, as $\xi^j_0 = \xi^{\rm S} = \xi^0_0$, the length-$N$ trajectory $\hat{u}(0) = u^0_{[0,N-1]}$, $\hat{y}(0) = y^0_{[0,N-1]}$ taken from $\{u^0,y^0\}$ is a feasible solution.
    Next, suppose $u^*(t),y^*(t),g^*(t)$ is the solution of \eqref{eq:nom_deeprc} at time step $t$. 
    At time step $t+1$, let $\hat{u}(t+1) = \col{u^*_{[1,N-1](t)},\bar{u}}$, $\hat{y}(t+1) = \col{y^*_{[1,N-1](t)},\bar{y}}$, where $\bar{u} = \pi(\xi^{*}_{N}(t)) \in \mathcal{U}$.
    According to Theorem \ref{theorem:invariance}, $(\hat{u}(t+1),\hat{y}(t+1))$ satisfies all constraints and there exists some $\hat{g}(t+1)$ satisfying \eqref{eq:nom_deeprc:model} since $(\hat{u}(t+1),\hat{y}(t+1))$ is a valid trajectory of the system.
    This proves recursive feasibility.
    \rzz{Using the candidate solution ${\hat{u}(t+1),\hat{y}(t+1)}$ and Theorem \ref{theorem:invariance}, we can prove exponential stability by following the arguments in \cite{berberich_data-driven_2021}.}
\end{proof}

As the Tube DeePRC \eqref{eq:deeprc_2s} solves a similar problem to the nominal DeePRC, it inherits similar closed-loop properties.
\begin{corollary}\label{corollary:robust_deeprc_rec_feas}
    \rzz{Under Under Assumptions \ref{ass:stand}, \ref{assumption:ini_robustly_safe_traj}, and \ref{assumption:quad_upperbound} (with $\xi$ replaced with $\zeta$ and \eqref{eq:nom_deeprc} replaced with \eqref{eq:deeprc_2s})} the following holds for all $j \geq 1$:
    \textit{\romannum{1}}) Problem \eqref{eq:deeprc_2s} is feasible for all $t \in \mathbb{Z}_{\geq 0}$;
    \textit{\romannum{2}}) $(v^j,z^j)$ satisfy the tightened constraints $\bar{\mathcal{U}}$, $\bar{\mathcal{Y}}$, and $\zeta_t^j$ \rzz{exponentially} converges to $\xi^{\rm F}$;
    \textit{\romannum{3}}) $(u^j,y^j)$ satisfy the nominal constraints $\mathcal{U}$, $\mathcal{Y}$, and $\xi_t^j\to$ \rz{$\xi^\text{F} + \mathcal{E}$} for all $d^j_t \in \mathcal{D}$.
\end{corollary}

\subsection{Recursive Feasibility of the End-to-End Formulation}
Let $\mathcal{D} = \{d \in \mathbb{R}^{m} \mid \|d\|_{\infty} \leq \bar{d}\}$ for some $\bar{d} \in \mathbb{R}_{>0}$, and let $\mathcal{K}^{j,t}_{\rm u} = \mathcal{K}^{j,t}_{[-m:,:]}$.
\begin{proposition}
    Under Assumption \ref{assumption:ini_robustly_safe_traj}, if $\bar{d} \geq \epsilon$ and the $1$-norm of each non-zero column of $\mathcal{K}^{j,t}_{\rm u}$ is $1$, then \eqref{eq:deeprc_1s} is recursively feasible.
\end{proposition}
\begin{proof}    
    Let ${\mathcal{K}}^{j,t}_{{\rm u}[:,i]}$ denote the $i$-th column of $\mathcal{K}^{j,t}_{\rm u}$, and suppose it is non-zero. Given any $\hat{u}_0(t)$, we have
    \begin{align}
        & \|\col{ y^j_{[t-\bar{L}+1,t-1]} , u^j_{[t-\bar{L}+1,t-1]} , \hat{u}_0(t) + \hat{d}(t) }^{\top} {\mathcal{K}}^{j,t}\|_{\infty} \notag \\
        = \, & \|c + \hat{d}(t)^{\top} {\mathcal{K}}^{j,t}_{{\rm u}} \|_{\infty} \geq |c_i + \hat{d}(t)^{\top} {\mathcal{K}}^{j,t}_{{\rm u}[:,i]} | \label{eq:deeprc_1s_rfeas_step1}
    \end{align}
    where $c$ is a constant row vector and $c_i$ denotes its $i$-th element. By choosing $\hat{d}(t) = \bar{d} \cdot {\rm sign}({\mathcal{K}}^{j,t}_{{\rm u}[:,i]})$ if $c_i \geq 0$ and $\hat{d}(t) = - \bar{d} \cdot {\rm sign}({\mathcal{K}}^{j,t}_{{\rm u}[:,i]})$ otherwise,
    with ${\rm sign}(\cdot)$ is applied element-wise, it holds that $\eqref{eq:deeprc_1s_rfeas_step1} \geq \bar{d} \cdot \|\mathcal{K}^{j,t}_{{\rm u}[:,i]}\|_1 \geq \epsilon$. In other words, for any $\hat{u}_{0}(t)$, we can construct the disturbance $\hat{d}(t)$ such that the excitation constraint \eqref{eq:deeprc_1s:excitation_condition} is satisfied. The rest of the proof constructs a candidate input/output sequence similar to that of Proposition \ref{proposition:closed-loop}, such that constraints \eqref{eq:deeprc_1s:compact_cnstr}-\eqref{eq:deeprc_1s:io_cnstr} are satisfied.
\end{proof}

\subsection{Convergence Properties}

From the discussion in Subsection \ref{subsection:LKB}, we immediately have the following.
\begin{proposition}\label{proposition:rank_convergence}
Both the DeePRC with LKB active exploration in Algorithm \ref{algo:deeprc_2s} and the end-to-end formulation \eqref{eq:deeprc_1s} reach the desired prediction horizon $\bar{N}$ in exactly $m\bar{L}+n-\operatorname{rank}(\HH_{\bar{N}}^1)$ time steps.
\end{proposition}

\rz{Proposition~\ref{proposition:rank_convergence} implies that we can obtain finite horizon approximations of \eqref{eq:ift_prob} in finitely many time steps and a direct data-driven fashion. 
Moreover, if $\xi_t$ in \eqref{eq:ift_prob} reaches the origin in finitely many time-steps $c$, and $c \leq \bar{N}$, then $(u^j,y^j)$ converges to the solution of \eqref{eq:ift_prob} in finitely many iterations. }
Using the arguments in~\cite{rosolia_optimality_2023}, we further have the following.
\begin{theorem} \label{proposition:convergence_optimal_inf}
Let Assumption \ref{assumption:ini_safe_traj} hold, and suppose the desired prediction horizon $\bar{N}$ is achieved at iteration $\bar{j}$.
Then $J^j(\xi^{\rm S}) \geq J^{j+1}(\xi^{\rm S})$ for all $j > \bar{j}$.
Moreover, if $u^{j+1} = u^{j} $ in Algorithm~\ref{algo:deeprc_2s} for some finite $j$ and the LICQ condition in~\cite[Assumption 3]{rosolia_optimality_2023} holds, then $u^j$ solves~\eqref{eq:ift_prob}.
\end{theorem}
\begin{proof}
The first point follows from \cite[Theorem 2]{rosolia_learning_2018}. The second point follows from~\cite[Theorem 2]{rosolia_optimality_2023}.
\end{proof}

For $j \leq \bar{j}$, we can also prove that the nominal cost of \eqref{eq:deeprc_2s} is non-increasing over iterations.
\rz{Theorem \ref{proposition:convergence_optimal_inf} guarantees that the solution of \eqref{eq:ift_prob} be reached asymptotically for any horizon $N$. However, controllers with larger horizons generally lead to faster convergence, as showcased in the simulations of Section \ref{sec:numerical_examples}, and exhibit overall better performance.}

\section{Numerical Examples} \label{sec:numerical_examples}
We consider the following four-dimensional system
\begin{equation*}
\begin{aligned}
    &\resizebox{\columnwidth}{!}{$A  = \begin{bmatrix}
        1 & 0 & 0.1 & 0 \\
        0 & 1 & 0 & 0.1 \\
        0 & 0 & 1 & 0 \\
        0 & 0 & 0 & 1
    \end{bmatrix}\!,
    B = \begin{bmatrix}0.1 & 0 \\ 0 & 0.1 \\ 0 & 0.1 \\ 0.1 & 0\end{bmatrix}\!, C  = \begin{bmatrix}1 & 0 & 0 & 0 \\ 0 & 1 & 0 & 0\end{bmatrix}$\!,} 
\end{aligned}
\end{equation*}
with $D=0$, subject to $\mathcal{U} = [-1.5,1.5]^{2}$, $\mathcal{Y} = [-1.5,1.5]^{2}$. We assume the true lag $\underline{\ell}=2$ is unknown and set $\ell = 4$. The initial and target states are $(u^{\rm S}, y^{\rm S}) = ([0 \ 0]^{\top}, [0 \ 0]^{\top})$ and $(u^{\rm F}, y^{\rm F}) = ([0 \ 0]^{\top}, [0.4 \ -\!0.4]^{\top})$. 
The cost matrices are $Q=\mathbf{I}^4$ and $R=0.1 \cdot \mathbf{I}^2$. 
In this setting, the desired prediction horizon is $\bar{N}=50$.
The initial safe trajectory $\{u^0,y^0\}$ is obtained using the LQR controller ($\bar{Q} = \mathbf{I}^{4}, \bar{R} = \mathbf{I}^2$). To meet Assumption \ref{assumption:ini_robustly_safe_traj}, small random disturbances are added to the LQR control input for the first $20$ time steps.
Using \eqref{eq:max_N}, we verify that $\{u^0,y^0\}$ produces a prediction horizon of $N^1 = 8$. 

In Figure \ref{fig:rank_increment}, we compare the performance of passive Hankel update (Passive), two-stage with LKB design (2s-LKB), \rz{end-to-end} approach (end-to-end).
The figure illustrates the relationship between the rank of the matrix $\HH_{\bar{N}}^j$ and its number of columns. 
Without exploration, the rank increment is considerably slower, indicating that many added trajectories to the matrix lack sufficient information. 
Meanwhile, for the methods with exploration, the trend is approximately linear, as expected.
\begin{figure}[b]
    \centering
    \input{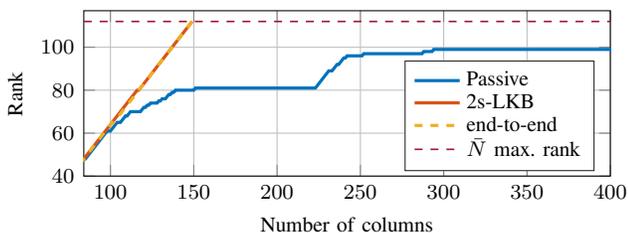}
    \vspace{-0.25cm}
    \caption{Evolution of the rank of the Hankel matrix.}
    \label{fig:rank_increment}
\end{figure}
\begingroup\addtolength{\tabcolsep}{-1pt}
\begin{table}[b]
\centering
\caption{Prediction horizon and trajectory costs over iterations}
\begin{tabular}{|c|cc|cc|cc|cc|}
\hline
 & \multicolumn{2}{c|}{Nominal} & \multicolumn{2}{c|}{Passive} & \multicolumn{2}{c|}{2s-LKB} & \multicolumn{2}{c|}{End-to-end} \\ \cline{2-9} 
\multirow{-2}{*}{$j$} &  $N$ & $J^j$  &  $N$ & $J^j$ &  $N$ & $J^j$ & $N$ & $J^j$ \\ \hline
0 & \textbackslash{} & 8.635158 & \textbackslash{} & 8.635158 & \textbackslash{} & 8.635158 & \textbackslash{} & 8.635158 \\
1 & 8 & 7.784081 & 8 & 7.818932 & 8 & 7.827334 & 8 & 7.826908 \\
2 & 8 & 7.765810 & 17 & 7.756467 & 50 & 7.748497 & 50 & 7.748497 \\
3 & 8 & 7.761973 & 19 & 7.749598 & 50 & 7.748497 & 50 & 7.748497 \\
4 & 8 & 7.760775 & 20 & 7.748698 & 50 & 7.748497 & 50 & 7.748497 \\
\hline
\end{tabular}
\label{tab:four_tank_deeprc}
\end{table}\endgroup
In Table \ref{tab:four_tank_deeprc} we further compare the closed-loop costs of the 2-stage LKB and the 1-stage approaches against a nominal DeePRC scheme with fixed horizon $N=8$ (without exploration) and againts the passive scheme.
The methods with active exploration achieve worse cost in $\textstyle{j=1}$ due to constraint tightening, but outperform the nominal and the passive schemes later thanks to the extended prediction horizon.
In the first iteration, the \rz{end-to-end} performs better than the LKB, since disturbance design is incoroprated in the MPC problem, \rz{however, it requires an average computation time of 0.757s, as opposed to the 0.058 seconds of the two-stage methods}.
\rz{The difference in performance between active and passive methods in our example is small, highlighting that further research is required to understand the impact of prediction horizon on closed-loop performance. Indeed, Theorem \ref{proposition:closed-loop} guarantees that the desired prediction horizon will be reached through active exploration, but provides no claims about the relative performance improvement.}

\section{Conclusion} \label{sec:conclusion}
In this work, we presented the DeePRC algorithm, a direct data-driven approach to control LTI systems performing iterative tasks. 
The DeePRC requires only an initial safe trajectory with a low excitation order and it guarantees sample-efficient rank-increment while maintaining safety in closed-loop.
Under suitable conditions, DeePRC is guaranteed to converge to the best achievable performance.
Extending the formal guarantees for systems with process disturbances, measurement noise, nonlinearities, \rz{and non-identical initial conditions} are promising directions for future work.

\bibliographystyle{IEEEtran}
\bibliography{Sources/ref.bib}

\end{document}